\documentclass[12pt,draft,onecolumn]{IEEEtran}
\usepackage{times,color,amsmath,amssymb,graphicx,epsf,psfrag,cite}

\def\boxit#1{\vbox{\hrule\hbox{\vrule\kern3pt
        \vbox{\kern3pt#1\kern3pt}\kern3pt\vrule}\hrule}}
\def\dag{{\cal y}}
\def\reals{ { {\rm  I \kern-0.15em R }  } }
\def\complex{ {\,{{\rm C} \kern-0.50em \raise0.20ex {  |}}\, }}

\def\Rbf{{\bf R}}

\def\Tc{{\cal T}}

\def\Rxx{\Rbf_{\ssstyle X\kern-.1em X}}

\let\ssstyle=\scriptscriptstyle

\def\Kout{\setbox1=\hbox{\Huge\bf K}\hbox to
1.05\wd1{\hspace{.05\wd1}
}}
\def\Sout{\setbox1=\hbox{\Huge\bf S}\hbox to 1.05\wd1{\hspace{.05\wd1}
}}

\def\tcm{\textcolor{blue}}
\def\tcr{\textcolor{blue}}
\def\tcb{\textcolor{blue}}

\def\scalefig#1{\epsfxsize #1\textwidth}
\def\nn{{\nonumber}}

\newtheorem{lemma}{Lemma}
\newtheorem{theorem}{Theorem}

\newtheorem{thm}{Theorem}
\newtheorem{prop}[thm]{Proposition}

\newtheorem{cor}{Corollary}

\newtheorem{ex}{{\em Example}}
\newtheorem{assumption}{Assumption}

\renewcommand{\theenumi}{\roman{enumi}}
\date{May 6, 2008}

\newcommand{\ignore}[1]{}
\begin{document}

\title{Optimality of Myopic Sensing in\\  Multichannel Opportunistic Access\thanks{Preliminary version of this work was presented at
\em{IEEE International Conference on Communications (ICC)}, May 2008, Beijing, China.} }
\author{
Sahand Haji Ali Ahmad$^\sharp$,
Mingyan Liu$^\sharp$
Tara Javidi$^\dag$,
Qing Zhao$^\ddag$,\\
Bhaskar Krishnamachari$^\S$
\\
\small
shajiali@eecs.umich.edu,
mingyan@eecs.umich.edu,
tara@ece.ucsd.edu,
qzhao@ece.ucdavis.edu,
bkrishna@usc.edu
\\
$^\sharp$Department of Electrical Engineering and Computer Science,
University of Michigan, Ann Arbor, MI 48109\\
$^\dag$Department of Electrical and Computer Engineering, University of California, San Diego, La Jolla, CA 92093\\
$^\ddag$Department of Electrical and Computer Engineering, University of California, Davis, CA 95616\\
$^\S$Ming Hsieh Department of Electrical Engineering, University of
Southern California, Los Angeles, CA 90089
}

\maketitle

\begin{abstract}
We consider opportunistic communication over multiple channels where the state (``good'' or ``bad'') of each channel evolves as independent and identically distributed Markov processes. A user, with limited channel sensing and access capability, chooses one channel to sense and subsequently access (based on the sensed channel state) in each time slot. A reward is obtained whenever the user senses and accesses a ``good'' channel. The objective is to design an optimal channel selection policy that maximizes the expected total (discounted or average) reward accrued over a finite or infinite horizon. This problem can be cast as
a Partially Observable Markov Decision Process (POMDP) or a restless multi-armed bandit
process, to which optimal solutions are often intractable. 
\tcb{
We show in this paper that a myopic policy that maximizes the immediate one-step reward is always optimal when the state transitions are positively correlated over time.  When the state transitions are negatively correlated, we show that the same policy is optimal when the number of channels is limited to 2 or 3, while presenting a counterexample for the case of 4 channels. } 
This result finds applications in opportunistic transmission scheduling in a fading environment, cognitive radio networks for spectrum overlay, and resource-constrained jamming and anti-jamming.
\end{abstract}

\vspace{1em}


\newpage

\section{Introduction}\label{sec:intro}


We consider a communication system in which a sender
has access to multiple channels, but is limited to sensing
and transmitting only on one at a given time. We explore how a smart
sender should exploit past observations and the knowledge of the
stochastic state evolution of
these channels to maximize its transmission rate by switching
opportunistically across channels.

We model this problem in the following manner.
As shown in Figure~1,
there are $n$
channels, each of which evolves as an independent,
identically-distributed, two-state discrete-time Markov chain. The
two states for each channel --- ``good" (or state $1$) and ``bad" (or
state $0$) ---
indicate the desirability of transmitting over that channel at a given time slot.
The state transition probabilities are given by $p_{ij}$, $i, j = 0, 1$.
In each time slot the sender picks one of the channels to sense based
on its prior observations, and obtains some fixed reward if it is in
the good state. The basic objective of the sender is to maximize the
reward that it can gain over a given finite time horizon. This
problem can be described as a partially observable Markov decision
process (POMDP)~\cite{Smallwood&Sondik:71OR} since the states
of the underlying Markov chains are not fully observed.
\tcb{It can also be cast as a special case of the class of
restless multi-armed bandit problems \cite{whittle}; more discussion on this 
is given in Section \ref{sec:discussion}. }

\begin{figure}[htb]
\centerline{
\begin{psfrags}
\psfrag{A}[c]{ $0$}
\psfrag{B}[c]{ $1$}
\psfrag{A1}[c]{ (bad)}
\psfrag{B1}[c]{ (good)}
\psfrag{a}[c]{$p_{01}$}
\psfrag{b}[l]{ $p_{11}$}
\psfrag{a1}[r]{$p_{00}$}
\psfrag{b1}[c]{ $p_{10}$}
\scalefig{0.5}\epsfbox{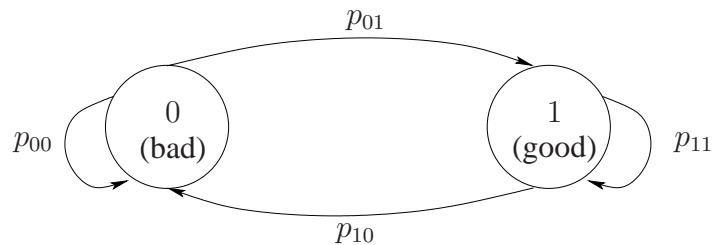}
\end{psfrags}}
\caption{The Markov channel model.} \label{fig:MC}
\end{figure}

This formulation is broadly applicable to several domains. It arises
naturally in opportunistic spectrum access
(OSA)~\cite{Zhao&Sadler:07SPM,Zhao&etal:07JSAC}, where the sender is a
secondary user, and the channel states describe the occupancy by
primary users. In the OSA problem, the secondary sender may send on
a given channel only when there is no primary user occupying it. It
pertains to communication over parallel fading channels as well, if
a two-state Markovian fading model is employed. Another interesting
application of this formulation is in the domain of communication
security, where it can be used to \color{blue} develop \color{black} bounds on the
performance of resource-constrained jamming. A jammer that has
access to only one channel at a time could also use the same
stochastic dynamic decision making process to maximize the number of
times that it can successfully jam communications that occur on
these channels. In this application, the ``good" state for the
jammer is precisely when the channel is being utilized by other
senders (in contrast with the OSA problem).


\tcm{In this paper we examine the optimality of a simple myopic policy
for the opportunistic access problem outlined above. 
}
\tcb{Specifically, we show that the myopic policy is optimal for arbitrary $n$ when
$p_{11}\geq p_{01}$.  We also show that it is optimal for $n=3$
when $p_{11}<p_{01}$, while presenting a finite horizon counter example showing that
it is in general not optimal for $n\geq 4$. 
We also generalize these results to related formulations involving discounted and average rewards over an infinite horizon.}


\tcm{These results extend and complement those reported in prior work~\cite{Zhao&etal:08TWC}. 
Specifically, it has been shown in~\cite{Zhao&etal:08TWC}
that for all $n$ the myopic policy has an elegant and robust
structure that obviates the need to know the channel state
transition probabilities and reduces channel selection to a simple round robin
procedure. Based on this structure, the optimality of the myopic policy for $n=2$ was established and the performance of the myopic policy, in particular, the scaling property with respect to $n$, analyzed in~\cite{Zhao&etal:08TWC}.
}
\tcm{It was conjectured in~\cite{Zhao&etal:08TWC} that the myopic policy
is optimal for any $n$. This conjecture was partially addressed in a preliminary
conference version~\cite{icc-08}}, \tcb{where the optimality was established 
under certain restrictive conditions on the channel parameters and the discount factor.
In the present paper, we significantly relax these conditions and formerly prove this conjecture under the condition $p_{11}\geq p_{01}$.  We also provide a counter example for $p_{11} < p_{01}$.}

\tcb{We would like to emphasize that compared to earlier work \cite{Zhao&etal:08TWC,icc-08},
the approach used in this paper relies on a coupling argument, which is the key to extending the optimality result to the arbitrary $n$ case.
Earlier techniques were largely based on exploiting the \tcr{convex analytic} properties of the value function, and were shown to
have difficulty in overcoming the $n=2$ barrier \tcr{without further conditions on the discount factor or
transition probabilities}.  This observation is somewhat reminiscent of the results reported in \cite{modiano}, where a coupling argument
was also used to solve an $n$-queue problem while earlier versions \cite{Ganti} using value function properties were limited to a $2$-queue case.}
\tcr{We invite the interested reader to refer to \cite{koolNOW}, an important
manuscript on monotonicity in MDPs which explores the power as well as the
limitation of working with analytic properties of value functions and dynamic programming
operators as we had done in our earlier work. In particular, \cite[Section 9.5]{koolNOW} explores the difficulty of using such techniques
for multi-dimensional problems where the number of queues is more than $n=2$; \cite[Chapter 12]{koolNOW} contrasts this proof technique with the
stochastic coupling arguments, which our present work uses.}


The remainder of this paper is organized as follows.  We formulate the problem
in Section~\ref{sec:problem} and illustrate
the myopic policy in Section~\ref{sec:policy}.  In Section~\ref{sec:optimal},
we prove that the myopic policy is optimal
in the case of \color{blue} $p_{11}\geq p_{01}$, \color{black} and show in Section~\ref{sec:counter}
that it is in general not optimal when this condition does not hold.
Section~\ref{sec:infinite} extends the results from finite horizon
to infinite horizon.
We discuss our work within the context of the class of restless bandit
problems as well as some related work in this area in Section~\ref{sec:discussion}.
Section~\ref{sec:conclusion} concludes the paper.




\section{Problem Formulation}\label{sec:problem}

We consider the scenario where a user is trying to access the wireless
spectrum to maximize its throughput or data rate.  The spectrum consists
of $n$ independent and statistically identical channels.
The state of a channel is given by a two-state discrete time Markov
chain shown in Figure \ref{fig:MC}.

The system operates in discrete time steps indexed by $t$, $t=1, 2, \cdots, T$,
where $T$ is the time horizon of interest.  At time $t^{-}$, the channels (i.e.,
the Markov chains representing them) go through state transitions, and at
time $t$ the user makes the channel sensing and access decision.
Specifically, at time $t$ the user selects one of the $n$ channels to sense,
say channel $i$.  If the channel is sensed to be in the ``good'' state (state $1$), the user transmits and collects one unit of reward.  Otherwise
the user does not transmit (or transmits at a lower rate), collects no reward, and waits until $t+1$ to make another choice.
This process repeats sequentially until the time horizon expires.

\color{blue}
As mentioned earlier, this abstraction is primarily motivated by the following multi-channel access scenario where a secondary user seeks spectrum opportunity in between a primary user's activities.  Specifically, time is divided into frames and at the beginning of each frame there is a designated time slot for the primary user to reserve that frame and for secondary users to perform channel sensing.  If a primary user intends to use a frame it will simply remain active in a channel (or multiple channels) during that sensing time slot (i.e., reservation is by default for a primary user in use of the channel), in which case a secondary user will find the channel(s) busy and not attempt to use it for the duration of that frame.  If the primary user is inactive during this sensing time slot, then the remainder of the frame is open to secondary users.  Such a structure provides the necessary protection for the primary user as channel sensing (in particular active channel sensing that involves communication between a pair of users) conducted at arbitrary times can cause undesirable interference.

Within such a structure, a secondary user has a limited amount of time and capability to perform channel sensing, and may only be able to sense
one or a subset of the channels before the sensing time slot ends.  And if all these channels are unavailable then it will have to wait till the
next sensing time slot.  In this paper we will limit our attend to the special case where the secondary user only has the resources to sense
one channel within this slot.  Conceptually our formulation is easily extended to the case where the secondary user can sense multiple channels
at a time within this structure, although the corresponding results differ, see e.g., \cite{Liu&Zhao:08Asilomar}.

Note that in this formulation we do not explicitly model the cost of channel sensing; it is implicit in the fact that the user is limited in how many channels it can sense at a time.  Alternative formulations have been studied where sensing costs are explicitly taken into consideration in a user's sensing and access decision, see e.g., a sequential channel sensing scheme in \cite{chang-mobicom07}.
\color{black}

In this formulation we have assumed that sensing errors are negligible.
Techniques used in this paper may be applicable in proving the optimality of
the myopic policy under imperfect sensing and for a general number of channels. The reason behind this is that our proof exploits the simple structure of the myopic policy, which remains when sensing is subject to errors as shown in~\cite{CrownCom}.

\color{blue}
Note that the system is not fully observable to the user, i.e., the user does not know the exact state of the system when making the sensing decision.  Specifically, channels go through state transition at time $t^{-}$ (or anytime between $(t-1, t)$), thus when the user makes the channel sensing decision at time $t$, it does not have the true state of the system at time $t$, which we denote by  ${\bf s}(t) = [s_1(t), s_2(t), \cdots, s_n(t)] \in \{0, 1\}^n$.  Furthermore, even after its action (at time $t^{+}$) it only gets to observe the true state of one channel, which goes through another transition at or before time $(t+1)^{-}$.
The user's action space at time $t$ is given by the finite set $\{1, 2, \cdots, n\}$, and we will use $a(t) = i$ to denote that the user selects channel $i$ to sense at time $t$.
For clarity, we will denote the outcome/observation of channel sensing at time $t$ following the action $a(t)$ by $h_{a(t)}(t)$, which is essentially the true state $s_{a(t)}(t)$ of channel $a(t)$ at time $t$ since we assume channel sensing to be error-free.

\color{black}
It can be shown (see e.g., \cite{Smallwood&Sondik:71OR,marcus,kumar}) that a sufficient statistic of such a system for optimal decision making, or the {\em information state} of the system \cite{marcus,kumar}, is given by the conditional probabilities of the state each channel is in given all past actions and observations.  Since each channel can be in one of two states, we denote this information state or belief vector by ${\bar\omega}(t)=[\omega_1(t),\cdots,\omega_n(t)] \in [0, 1]^n$, where $\omega_i(t)$ is the conditional probability that channel $i$ is in state $1$ at time $t$ given all past states, actions and observations \footnote{Note that this is a standard way of turning a POMDP problem into a classic MDP (Markov decision process) problem by means of
information state, the main implication being that the state space is
now uncountable.}.
\color{blue} Throughout the paper $\omega_i(t)$ will be referred to as the information state of channel $i$ at time $t$, or simply the channel probability of $i$ at time $t$. \color{black}

\color{blue}
Due to the Markovian nature of the channel model, the future information state is only a function of the current information state and the current action; i.e., it is independent of past history given the current information state and action.
It follows that the information state of the system evolves as follows.  Given that the state at time $t$ is ${\bar\omega}(t)$ and action $a(t)=i$ is taken, $\omega_i(t+1)$ can take on two values: (1) $p_{11}$ if the observation is that channel $i$ is in a ``good'' state ($h_i(t)=1$); this occurs with probability $P\{h_i(t)=1|{\bar\omega}(t)\} = \omega_i(t)$; (2) $p_{01}$ if the observation is that channel $i$ is in a ``bad'' state ($h_i(t)=0$); this occurs with probability $P\{h_i(t)=0 | {\bar\omega}(t)\} = 1-\omega_i$.  For any other channel $j\neq i$, the corresponding $\omega_j(t+1)$ can only take on one value (i.e., with probability $1$): $\omega_j(t+1) = \tau(\omega_j(t))$ where the operator $\tau: [0,1] \rightarrow [0,1]$ is defined as
\begin{equation}
\tau(\omega) := \omega p_{11} + (1-\omega) p_{01},~~~ 0\le\omega\le 1.
\end{equation}

These transition probabilities are summarized in the following equation for $t=1, 2, \cdots, T-1$:
\begin{eqnarray}
\{ \omega_i(t+1) | {\bar\omega}(t), a(t)  \}
= \left\{ \begin{array}{ll}
p_{11} & \mbox{with prob. } \omega_i(t) ~ \mbox{if } a(t)=i \\
p_{01} & \mbox{with prob. } 1-\omega_i(t) ~ \mbox{if } a(t)=i \\
\tau(\omega_i(t)) & \mbox{with prob. } 1 ~ \mbox{if } a(t) \neq i\\
\end{array}
\right.,  ~ i = 1, 2, \cdots, n, \label{eq:omega}
\end{eqnarray}
%
Also note that ${\bar\omega}(1)\in [0,1]^n$ denotes the initial condition (information state) of the system, which may be interpreted as the user's initial belief about how likely each channel is in the good state before sensing starts at time $t=1$.  For the purpose of the optimization problems formulated below, this initial condition is considered given, which can be any probability vector \footnote{\color{blue}That is, the optimal solutions are functions of the initial condition.  A reasonable choice, if the user has no special information other than the transition probabilities of these channels, is to simply use the steady-state probabilities of channels being in state ``1'' as an initial condition (i.e., setting $\omega_i(1)=\frac{p_{10}}{p_{01}+p_{10}}$).\color{black}}.

\tcb{It is important to note that although in general a POMDP problem has an uncountable state space (information states are probability distributions), in our problem the state space is countable for any given initial condition $\bar\omega(1)$.  This is because as shown above, the information state of any channel with an initial probability of $\omega$ can only take on the values $\{\omega, \tau^{k}(\omega), p_{01}, \tau^{k}(\omega), p_{11}, \tau^{k}(\omega)\}$, where $k=1, 2, \cdots$ and $\tau^{k}(\omega):=\tau(\tau^{k-1}(\omega))$, which is a countable set.}

For compactness of presentation we will further use the operator $\Tc$ to denote the above probability distribution of the information state (the entire vector):
\begin{eqnarray}\label{eq:T}
\bar\omega(t+1) = \Tc(\bar\omega(t), a(t)),
\end{eqnarray}
by noting that the operation given in (\ref{eq:omega}) is applied to $\bar{\omega}(t)$ element-by-element.  We will also use the following to denote the information state given observation outcome:
\begin{eqnarray}
\Tc(\bar\omega(t), a(t) | h_{a(t)}(t)=1)
= (\tau(\omega_1(t)), \cdots, \tau(\omega_{a(t)-1}(t)), p_{11},
\tau(\omega_{a(t)+1}(t)), \cdots, \tau(\omega_n(t))) \\
\Tc(\bar\omega(t), a(t) | h_{a(t)}(t)=0)
= (\tau(\omega_1(t)), \cdots, \tau(\omega_{a(t)-1}(t)), p_{01},
\tau(\omega_{a(t)+1}(t)), \cdots, \tau(\omega_n(t)))
\end{eqnarray}

\color{black}

%

The objective of the user is to maximize its total (discounted or average) expected
reward over a finite (or infinite) horizon.  Let $J_T^{\pi}({\bar\omega})$,
$J^{\pi}_{\beta}({\bar\omega})$, and $J^{\pi}_{\infty}({\bar\omega})$ denote, respectively,
these cost criteria (namely, finite horizon, infinite horizon with discount, and
infinite horizon average reward) under policy $\pi$ starting in state ${\bar\omega} = [\omega_1, \cdots, \omega_n]$.
The associated optimization problems ((P1)-(P3)) are formally defined
as follows.
\begin{eqnarray}
\mbox{(P1):} && \hspace{-0.2in} \max_{\pi} J_T^{\pi}({\bar\omega})
= \max_{\pi} E^{\pi} [ \sum_{t=1}^{T} \beta^{t-1} R_{\pi_t}({\bar\omega}(t)) | {\bar\omega}(1) = {\bar\omega}]\nn \\
\mbox{(P2):} && \hspace{-0.2in} \max_{\pi} J^{\pi}_{\beta}({\bar\omega})
= \max_{\pi} E^{\pi} [\sum_{t=1}^{\infty} \beta^{t-1} R_{\pi_t}({\bar\omega}(t)) | {\bar\omega}(1) = {\bar\omega}]\nn\\
\color{blue}
\mbox{(P3):} &&
\color{blue}
\hspace{-0.2in} \max_{\pi} J^{\pi}_{\infty}({\bar\omega})
= \max_{\pi} \lim_{T\rightarrow\infty} \frac{1}{T}E^{\pi} [\sum_{t=1}^{T} R_{\pi_t}({\bar\omega}(t)) | {\bar\omega}(1) = {\bar\omega}]\nn
\color{black}
\end{eqnarray}
where $\beta$ ($0\leq \beta \leq 1$ for (P1) and $0\leq \beta < 1$ for (P2))
is the discount factor, and $R_{\pi_t}({\bar\omega}(t))$ is the reward collected under state ${\bar\omega}(t)$ when channel $a(t)=\pi_t(\bar\omega(t))$ is selected
\color{blue}
and $h_{a(t)}(t)$ is observed.  This reward is given by
$R_{\pi_t}({\bar\omega}(t)) = 1$ with probability $\omega_{a(t)}(t)$ (when $h_{a(t)}(t)=1$), and $0$ otherwise.

The maximization in (P1) is over the class of deterministic Markov
policies.\footnote{\color{blue}A Markov policy is a policy that derives its action only depending on the current (information) state, rather than the entire history of states, see e.g., \cite{kumar}.\color{black}}.  An admissible policy $\pi$, given by the vector $\pi = [\pi_1, \pi_2, \cdots, \pi_T]$, is thus such that $\pi_t$ specifies a mapping from the current information state
$\bar\omega(t)$ to a channel selection action $a(t)=\pi_t(\bar\omega(t)) \in\{1, 2, \cdots, n\}$.
This is done without loss of optimality due to the Markovian nature of the underlying system, and due to known results on POMDPs.  Note that the class of Markov policies in terms
of information state are also known as seperated policies (see \cite{kumar}). Due to
finiteness of (unobservable) state spaces and action space in problem~(P1), it is known that an optimal policy (over all random and deterministic, history-dependent and history-independent policies) may be found within the class of separated (i.e.\ deterministic Markov) policies (see e.g., \cite[Theorem 7.1, Chapter 6]{kumar}), thus justifying the maximization and the admissible policy space.

\tcr{In Section~\ref{sec:infinite} we establish the existence of a stationary separated policy $\pi^*$, under which the supremum of the expected discounted reward as well as the supremum of expected average cost are achieved, hence justifying our use of maximization in (P2) and (P3).
Furthermore, it is shown that under this policy the limit in (P3) exists and is greater than the limsup of the average performance of any other policy (in general history-dependent and randomized).  This is a strong notion of optimality; the interpretation is that the most ``pessimistic" average performance under policy $\pi^*$ ($\lim\inf \frac{1}{T}J_T^{\pi^*}(\cdot)=\lim
\frac{1}{T}J_T^{\pi^*}(\cdot)$) is greater than the most ``optimistic" performance under any
other policy $\pi$ ($\lim\sup \frac{1}{T}J_T^{\pi}(\cdot)$). In much of the literature on MDP, this
is referred to as the \emph{strong optimality} for an expected average cost (reward) problem; for a discussion on this, see \cite[Page 344]{puterman}.
}


\section{Optimal Policy and the Myopic Policy}\label{sec:policy}

\subsection{Dynamic Programming Representations}\label{sec:prelim}

\color{blue}
Problems (P1)-(P3) defined in the previous section may be solved using their respective dynamic programming (DP) representations.  Specifically, for problem (P1), we have the following recursive equations:
\begin{eqnarray}
V_{T}({\bar\omega}) &=& \max_{a=1, 2, \cdots, n} E[R_a({\bar\omega}) ] \nonumber\\
V_{t}({\bar\omega}) &=& \max_{a=1, 2, \cdots, n} E[R_a({\bar\omega}) + \beta V_{t+1}(\Tc({\bar\omega}, a))] \nn \\
&=& \max_{a=1,\cdots, n} ( \omega_a +\beta
\omega_a V_{t+1} \left( \Tc\left( {\bar\omega}, a | 1\right)\right)  
+\beta (1-\omega_a)  V_{t+1} \left( \Tc\left( {\bar\omega}, a | 0\right) \right) )~,
\label{DP-finite-t}
\end{eqnarray}
for $t=1, 2, \cdots, T-1$, where $V_t(\bar\omega)$ is known as the value function,  or the maximum expected future reward that can be accrued starting from time $t$ when the information state is ${\bar\omega}$.  In particular, we have $V_1(\bar\omega) = \max_{\pi} J^{\pi}_T(\bar\omega)$, and an optimal deterministic Markov policy exists such that $a=\pi^*_t (\bar\omega)$ achieves the maximum in (\ref{DP-finite-t}) (see e.g., \cite{puterman} (Chapter 4)).  Note that since $\cal{T}$ is a conditional probability distribution (given in (\ref{eq:T})), $V_{t+1}({\cal T}(\bar\omega, a))$ is taken to be the expectation over this distribution when its argument is ${\cal T}$, with a slight abuse of notation, as expressed in (\ref{DP-finite-t}).


%

Similar dynamic programming representations hold for (P2) and (P3) as given below.
\ignore{
As mentioned earlier, the existence of optimal stationary Markov policies in (P2) and (P3)
is a consequence of i) the finiteness of the action space, ii) the countability of the information state space for given initial condition,  and iii) the bounded condition on the immediate reward (see \cite{puterman} and \cite{marcus}) ({\bf Tara, again please see if anything here needs to be modified to be consistent with the previous section}).}
For problem (P2) there exists a unique function $V_\beta(\cdot)$ satisfying the following fixed point equation:
\begin{eqnarray}
V_{\beta}({\bar\omega}) &=& \max_{a=1, \cdots, n} E[R_a({\bar\omega}) + \beta V_{\beta}(\Tc({\bar\omega}, a))] \nn \\
 &=& \max_{a=1,\cdots, n} ( \omega_a + \beta
\omega_a V_\beta \left( \Tc\left( \bar\omega, a | 1\right)\right) 
+ \beta (1-\omega_a)  V_\beta \left( \Tc\left( \bar\omega, a | 0\right) \right) )~.
\label{DP-discount}
\end{eqnarray}
We have that $V_\beta(\bar\omega) = \max_{\pi} J_{\beta}^{\pi}({\bar\omega})$,  
and that a stationary separated policy $\pi^*$ is optimal
if and only if  $a=\pi^*(\bar\omega)$ achieves the maximum in (\ref{DP-discount})
\cite[Theorem 7.1]{marcus-survey}.

For problem (P3), we will show that there exist a bounded function $h_\infty(\cdot)$ and a constant scalar $J$
satisfying the following equation:
\begin{eqnarray}
J+ h_{\infty}({\bar\omega}) &=& \max_{a=1, 2, \cdots, n} E[R_a({\bar\omega}) +  h_{\infty}(\Tc({\bar\omega}, a))] \nn \\
&=& \max_{a=1,\cdots, n} ( \omega_a +
\omega_a h_{\infty} \left( \Tc\left( \bar\omega, a | 1\right)\right) 
+(1-\omega_a)  h_{\infty} \left( \Tc\left( \bar\omega, a | 0\right) \right) ).\label{DP-avgCost}
\end{eqnarray}
The boundedness of $h_\infty$ and the immediate reward implies
that $J = \max_{\pi} J_{\infty}^{\pi}({\bar\omega})$,
and that a stationary separated policy $\pi^*$ is optimal
in the context of (P3)
if and only if  $a=\pi^*(\bar\omega)$ achieves the maximum in
(\ref{DP-avgCost}) \cite[Theorems 6.1-6.3]{marcus-survey}.


Solving (P1)-(P3) using the above recursive equations is in general computationally heavy.
Therefore, instead of directly using the DP equations, the focus of this paper is on examining the optimality properties of a simple, greedy algorithm.
We define this algorithm next and show its simplicity in structure and implementation. 
\color{black}

\subsection{The Myopic Policy}\label{sec:policy-myopic}

A myopic or greedy policy ignores the impact of the current action on the future reward,
focusing solely on maximizing the expected immediate reward. Myopic policies
are thus stationary.
For (P1), the myopic policy under state ${\bar\omega} = [\omega_1, \omega_2,
\cdots, \omega_n]$ is given by
\begin{equation}
a^*({\bar\omega})=\arg\max_{a=1,\cdots,n} E[R_a({\bar\omega})] =
\arg\max_{a=1, \cdots, n}  \omega_a.
\label{eq:a*}
\end{equation}

In general, obtaining the myopic action in each time slot requires
the successive update of the information state as given in
\eqref{eq:omega}, which explicitly relies on the knowledge of the transition
probabilities $\{p_{ij}\}$ \color{blue}as well as the initial condition $\bar\omega(1)$.
\color{black}
Interestingly, \tcm{it has been shown in~\cite{Zhao&etal:08TWC}} that the implementation of the myopic policy requires only the knowledge of the initial condition and the order of $p_{11}$ and $p_{01}$, but not the precise values of these
transition probabilities.  \tcb{To make the present paper self-contained, below we briefly describe how this policy works; more details may be found in \cite{Zhao&etal:08TWC}.} 

\tcb{Specifically, when $p_{11}\geq p_{01}$
%
%
the conditional probability updating function $\tau(\omega)$ is a monotonically increasing function,  i.e., $\tau(\omega_1) \geq \tau(\omega_2)$
for $\omega_1 \geq \omega_2$.  Therefore the ordering of information states among channels is preserved when they are not observed.
If a channel has been observed to be in state ``1'' (respectively ``0''), its probability at the next step becomes $p_{11}\geq \tau(\omega)$
(respectively $p_{01}\leq \tau(\omega)$) for any $\omega\in[0, 1]$.  In other words, a channel observed to be in state ``1'' (respectively ``0'')
will have the highest (respectively lowest) possible information state among all channels.}
%



\tcb{These observations lead to the following implementation of the myopic policy.
We take the initial information state $\bar\omega(1)$, order the channels according to their probabilities $\omega_i(1)$,
and probe the highest one (top of the ordered list) with ties broken randomly.  In subsequent steps we stay in the same channel
if the channel was sensed to be in state ``1'' (good) in the previous slot; otherwise, this channel is moved to the bottom of the ordered list,
and we probe the channel currently at the top of the list.  This in effect creates a round robin style of probing, where the channels are
cycled through in a fixed order.
This circular structure is exploited in Section~\ref{sec:optimal} to prove the optimality of the myopic policy in the case of $p_{11}\geq p_{01}$.}

\tcb{When $p_{11}< p_{01}$, we have an analogous but opposite situation.
The conditional probability updating function $\tau(\omega)$ is now a monotonically decreasing function,  i.e., $\tau(\omega_1) \leq \tau(\omega_2)$
for $\omega_1 \geq \omega_2$.  Therefore the ordering of information states among channels is reversed at each time step when they are not observed.
If a channel has been observed to be in state ``1'' (respectively ``0''), its probability at the next step becomes $p_{11}\leq \tau(\omega)$
(respectively $p_{01}\geq \tau(\omega)$) for any $\omega\in[0, 1]$.  In other words, a channel observed to be in state ``1'' (respectively ``0'')
will have the lowest (respectively highest) possible information state among all channels.}

\tcb{As in the previous case, these similar observations lead to the following implementation. 
We take the initial information state $\bar\omega(1)$, order the channels according to their probabilities $\omega_i(1)$, and probe the highest
one (top of the ordered list) with ties broken randomly.  In each subsequent step, if the channel sensed in the previous step was in
state ``0'' (bad), we keep this channel at the top of the list but completely reverse the order of the remaining list, and we probe this channel.
If the channel sensed in the previous step was in state ``1'' (good), then we completely reverse the order of the entire list (including dropping
this channel to the bottom of the list), and probe the channel currently at the top of the list.  This alternating circular structure is exploited in Section~\ref{sec:counter} to examine the optimality of the myopic policy in the case of $p_{11}<p_{01}$.}

\section{Optimality of the Myopic Policy in the Case of $p_{11}\geq p_{01}$}
\label{sec:optimal}

In this section we show that the myopic policy, with a simple and
robust structure, is optimal when $p_{11} \geq p_{01}$.
We will first show this for the finite horizon discounted cost case,
and then extend the result to the infinite horizon case under both discounted
and average cost criteria in Section~\ref{sec:infinite}.

The main assumption is formally stated as follows.
\begin{assumption} \label{assumption}
The transition probabilities $p_{01}$ and $p_{11}$ are such that
\begin{eqnarray}\label{p11>p01}
p_{11} - p_{01} & \geq & 0.
\end{eqnarray}
\end{assumption}



The main theorem of this section is as follows.
\begin{theorem} \label{mainThm}
Consider Problem (P1).
\color{blue}
Define $V_t(\bar\omega; a):= E[R_a(\bar\omega) + \beta V_{t+1}({\cal T}(\bar\omega, a))]$, i.e., the value of the value function given in Eqn (\ref{DP-finite-t}) when action $a$ is taken at time $t$ followed by an optimal policy.
\color{black}
Under Assumption~\ref{assumption},
the myopic policy is optimal, i.e.  for $\forall t, 1 \leq t < T$,
and $\forall \bar\omega = [\omega_1, \cdots, \omega_n] \in [0, 1]^{n}$,
\begin{equation} \label{mainIneq}
\color{blue}
V_t(\bar\omega; a = j ) -
V_t(\bar\omega; a = i ) \geq 0,
\end{equation}
if $\omega_j \geq \omega_i$, \color{black}
for $i=1,\cdots,n$.
\end{theorem}

The proof of this theorem is based on backward induction on $t$: 
given the optimality of the myopic policy at times $t+1, t+2, \cdots, T$, we want to show that it is also optimal at time $t$.  This relies on a number of lemmas introduced below.
The first lemma introduces a notation that allows us to express the expected future reward under the myopic policy.

\begin{lemma}\label{lem:T}
There exist $T$ $n$-variable functions, denoted by $W_t()$, $t=1, 2, \cdots, T$, each of which is a polynomial of order 1\footnote{Each function $W_t$ is \color{blue} affine in each variable, when all other variables are held constant. \color{black}} and can be represented recursively in the following form:
\begin{eqnarray}\label{eq:W}
W_t(\bar \omega) = \omega_n + \omega_n \beta W_{t+1}(\tau(\omega_1),\hdots , \tau(\omega_{n-1}) , p_{11})
+(1-\omega_n) \beta W_{t+1}(p_{01},\tau(\omega_1),\hdots , \tau(\omega_{n-1}) ),
\end{eqnarray}
where $\bar\omega=[\omega_1, \omega_2, \cdots, \omega_n]$ and
$W_T(\bar\omega) = \omega_n$.

\end{lemma}

\begin{proof}
The proof is easily obtained using backward induction on $t$ given the above recursive equation and noting
that $W_T()$ is one such polynomial and the mapping $\tau()$ is a linear operation.
\end{proof}

\begin{cor}\label{col:W}
\tcb{When $\bar\omega$ represents the ordered list of information states $[\omega_1, \omega_2, \cdots, \omega_n]$
with $\omega_1\le\omega_2\le \cdots\le\omega_n$, then $W_t(\bar\omega)$ is the expected total reward obtained by
the myopic policy from time $t$ on.}
\end{cor}

\tcb{This result follows directly from the description of the policy given in Section \ref{sec:policy-myopic}. }

\begin{prop}
The fact that $W_t$ is a polynomial of order 1 and affine in
each of its elements implies that
\begin{eqnarray}
&&W_t(\omega_1, \cdots, \omega_{n-2}, y, x) -
W_t(\omega_1,\cdots,\omega_{n-2}, x, y)\nn\\
& =&
(x-y)[W_t(\omega_1,\cdots,\omega_{n-2}, 0, 1)-
W_t(\omega_{1},\cdots,\omega_{n-2}, 1, 0)] ~.
\end{eqnarray}
Similar results hold when we change the positions of $x$ and $y$.
\end{prop}
To see this, consider $W_t(\omega_{1},\cdots,\omega_{n-2}, x, y)$ and
$W_t(\omega_{1},\cdots,\omega_{n-2}, y, x)$, as functions of $x$ and $y$,
each having an $x$ term, a $y$ term, an $xy$ term and a constant term.
Since we are just swapping the positions of $x$ and $y$ in
these two functions, the constant term remains the same,
and so does the $xy$ term. Thus the only difference is
the $x$ term and the $y$ term, as given in the above equation.
This linearity result will be used later in our proofs.

The next lemma establishes a necessary and sufficient condition
for the optimality of the myopic policy.
\begin{lemma}\label{lem:nes-suf}
Consider Problem (P1) and Assumption 1.
Given the optimality of \color{blue} the \color{black}
myopic policy at times  $t+1, t+2, \cdots, T$,
the optimality at time $t$ is equivalent to:
\begin{eqnarray}
W_t(\omega_1 , \hdots , \omega_{i-1} , \omega_{i+1} , \hdots , \omega_n , \omega_i) \leq  W_t(\omega_1, \hdots , \omega_n), \hspace*{.4in}
\mbox{ for all }  \omega_1  \leq \cdots \leq  \omega_i \leq \cdots \leq \omega_n. \nn
\end{eqnarray}
\end{lemma}

\begin{proof}
Since the myopic policy is optimal from $t+1$ on, it is sufficient to show that probing $\omega_n$ followed by myopic probing is better than probing any other channel followed by myopic probing.  The former is precisely given by the RHS of the above equation; the latter by the LHS, thus completing the proof.
\end{proof}

Having established that $W_t(\bar\omega)$ is the total expected reward of the myopic policy
for an increasingly-ordered vector $\bar\omega=[\omega_1, \cdots, \omega_n]$, we next proceed to show that we do not decrease
this total expected reward $W_t(\bar\omega)$ by switching the order of two neighboring elements $\omega_i$ and $\omega_{i+1}$
if $\omega_i\geq \omega_{i+1}$.  This is done in two separate cases, when $i+1<n$ (given in Lemma \ref{lem:exchange}) and
when $i+1=n$ (given in Lemma \ref{lem:exchange-border}), respectively.  The first case is quite straightforward, while proving
the second cased turned out to be significantly more difficult.
Our proof of the second case (Lemma \ref{lem:exchange-border}) relies on a separate lemma (Lemma \ref{lem:bound}) that establishes a bound \color{blue} between the greedy use of two identical vectors but with a different starting position.\color{black}
The proof of Lemma \ref{lem:bound} is based on a coupling argument and is quite instructive.
%
Below we \color{blue} present and prove Lemmas \ref{lem:bound}, \ref{lem:exchange} and \ref{lem:exchange-border}. \color{black}

\begin{lemma}\label{lem:bound}
For $0 < \omega_1 \leq \omega_2 \leq \hdots \leq \omega_n < 1$ , we have the following inequality for all $t=1, 2, \cdots, T$:
\begin{eqnarray}
1+W_t(\omega_2,\hdots,\omega_n,\omega_1)  \geq W_t(\omega_1,\hdots,\omega_n).
\end{eqnarray}
\end{lemma}

\color{blue}
\begin{proof}
We prove this lemma using a coupling argument along any sample path.
The LHS of the above inequality represents the expected reward of a policy
(referred to as L below) that probes in the sequence of channels $1$
followed by $n$, $n-1$, $\cdots$, and
then $1$ again, and so on, plus an extra reward of $1$;
the RHS represents the expected reward of a policy (referred to as R below)
that probes in the sequence of channels $n$ followed by $n-1$, $\cdots$, and $1$
and then $n$ again, and so on.
It helps to imagine lining up the $n$ channels along a circle in the sequence of $n$, $n-1$, $\cdots$, $1$, clock-wise, and thus L's starting position is $1$, R's starting position is $n$, exactly one spot ahead of L clock-wise.  Each will cycle around the circle till time $T$.

Now for any realization of the channel conditions (or any sample path of
the system), consider the sequence of ``$0$''s and
``$1$''s that these two policies see, and consider the position they are on the circle.  The reward a policy gets along a given sample path is
$R_l= \sum_{j=t}^{T}\beta^{j_l}$ for policy L, where $j_l = j$ if L sees a ``1'' at time $j$, and $0$ otherwise;
the reward for R is $R_r= \sum_{j=t}^{T}\beta^{j_r}$ with $j_r$ similarly defined.
There are two cases.

Case (1): the two eventually catch up with each other at some time $K\leq T$, i.e., at some point they
start probing exactly the same channel.  From this point on the two policies behave exactly the same way along the same sample path, and the reward they obtain from this point on is exactly the same.  Therefore in this case we only need to compare the rewards (L has an extra $1$) leading up to this point.

Case (2): The two never manage to meet within the horizon $T$.  In this case we need to compare the rewards for the entire horizon (from $t$ to $T$).

We will consider Case (1) first.  There are only two possibilities for the two policies to meet: (Case 1.a) either L has seen exactly one more ``0'' than R in its sequence, or (Case 1.b) R has seen exactly $n-1$ more ``0''s than L.  This is because the moment we see a ``0'' we will move to the next channel on the circle.  L is only one position behind R, so one more ``0'' will put it at exactly the same position as R.  The same with R moving $n-1$ more positions ahead to catch up with L.

Case (1.a): L sees exactly one more ``0'' than R in its sequence.
The extra ``0'' necessarily occurs at exactly time $K$, $t\leq K\leq T$, meaning that at $K$, L sees a ``0'' and R sees a ``1''.  From $t$ to $K$, if we write the sequence of rewards (zeros and ones) under L and R,
we observe the following: between $t$ and $K$ both L and R have equal number of zeros, while
for $\forall t' =t, t+1, \ldots, K-1$, the number of zeros up to time $t'$ is less (or no more) for L than for R. In other
words, L and R see the same number of ``0''s, but L's is always lagging behind (or no earlier).  That is, for every ``0'' R sees, L has a matching ``0'' that occurs no earlier than R's ``0."  This means that if we denote by
$R_l(t_1, t_2)$ the
rewards accumulated between $t_1$ and $t_2$, then for the rewards in $[t,K-1]$, we have $R_l(t, t') \geq R_r(t,t')$, for $\forall t'\leq K-1$, while $R_l(K, K)=\beta^K$ and
$R_r(K, K)=0$. Finally by definition we have $R_l(K+1,T )= R_r(K+1,T)$.
Therefore overall we have $1+R_l(t,T) \geq R_r(t,T)$, proving the above
inequality.

Case (1.b): R sees $n-1$ more ``0''s than L does. The comparison is simpler.  We only need to note that R's ``0''s must again precedes (or be no later than) L's since otherwise we will return to Case (1.a).
Therefore we have $R_l\geq R_r$, and thus $1+R_l \geq R_r$ is also true.

We now consider Case (2). The argument is essentially the same.  In this case the two don't get to meet, but they are on their way, meaning that either L has exactly the same ``0''s as R and their positions are no earlier (corresponding to Case (1.a)), or R has more ``0''s than L (but not up to $n-1$) and their positions are no later than L's (corresponding to Case (1.b)).
So either way we have $1+R_l \geq R_r$.

The proof is thus complete.
\end{proof}
\color{black}


\begin{lemma} \label{lem:exchange}
For all $j$, $1\leq j \leq n-3$, and all $x\geq y$, we have
\begin{eqnarray}
W_t(\omega_1,\hdots,\omega_{j},x,y,\hdots,\omega_n) \leq W_t(\omega_1,\hdots,\omega_{j},y,x,\hdots,\omega_n)
\end{eqnarray}
\end{lemma}

\begin{proof}
We prove this by induction over $t$. The claim is obviously true for $t=T$, since both sides will be equal to $\omega_n$, thereby establishing the induction basis.  Now suppose the claim is true for all $t+1, \cdots, T-1$.  We have
\begin{eqnarray}
&& W_t(\omega_1, \cdots, \omega_{j-1}, x, y, \cdots, \omega_n) \nonumber \\
&=& \omega_n (1+\beta W_{t+1}(\tau(\omega_1), \cdots, \tau(x), \tau(y), \cdots, \tau(\omega_{n-1}),
p_{11})) \nonumber \\
 &+& (1-\omega_n) \beta W_{t+1}(p_{01}, \tau(\omega_1), \cdots, \tau(x), \tau(y), \cdots, \tau(\omega_{n-1})) \nonumber \\
 &\leq& \omega_n (1+\beta W_{t+1}(\tau(\omega_1), \cdots, \tau(y), \tau(x), \cdots, \tau(\omega_{n-1}),
p_{11})) \nonumber \\
 &+& (1-\omega_n) \beta W_{t+1}(p_{01}, \tau(\omega_1), \cdots, \tau(y), \tau(x), \cdots, \tau(\omega_{n-1})) \nonumber \\
&=& W_t(\omega_1, \cdots, \omega_{j-1}, y, x, \cdots, \omega_n)
\end{eqnarray}
where the inequality is due to the induction hypothesis, and noting that $\tau()$ is a \color{blue} monotone increasing \color{black} mapping in the case of $p_{11}\geq p_{01}$.
\end{proof}

\begin{lemma} \label{lem:exchange-border}
For all  $x\geq y$, we have
\begin{eqnarray}
W_t(\omega_1,\hdots,\omega_{j},\hdots,\omega_{n-2}, x, y) \leq
W_t(\omega_1,\hdots,\omega_{j},\hdots,\omega_{n-2}, y,x).
\end{eqnarray}
\end{lemma}

\begin{proof}
This lemma is proved inductively.  The claim is obviously true for $t=T$.
Assume it also holds for times $t+1, \cdots, T-1$. We have by the definition of $W_t()$ and due to its linearity property:
\begin{eqnarray}
&& W_t(\omega_1,\hdots, \omega_{n-2},y,x) - W_t(\omega_1,\hdots, \omega_{n-2},x,y)  \nn \\
&=& (x - y)  (W_t(\omega_1 , \hdots , \omega_{n-2} , 0 ,1 )-W_t(\omega_1 , \hdots , \omega_{n-2} , 1, 0) )\nn\\
&=& (x - y) \left(
1+\beta W_{t+1}(\tau(\omega_1), \hdots, \tau(\omega_{n-2}) , p_{01} ,p_{11} )-
\beta W_{t+1}(p_{01},\tau(\omega_1), \hdots, \tau(\omega_{n-2}), p_{11}) \right).\nn
\end{eqnarray}
But from the induction hypothesis we know that
\begin{eqnarray}
 W_{t+1}(\tau(\omega_1) , \hdots ,\tau( \omega_{n-2}) , p_{01} ,p_{11} ) \geq
W_{t+1}(\tau(\omega_1), \hdots , \tau(\omega_{n-2}), p_{11}, p_{01}).
\end{eqnarray}
This means that
\begin{eqnarray}
&& 1+\beta W_{t+1}(\tau(\omega_1) , \hdots ,\tau( \omega_{n-2}) , p_{01} ,p_{11} )-
\beta W_{t+1}(p_{01},\tau(\omega_1), \hdots , \tau(\omega_{n-2}), p_{11}) \nn\\
&\geq &
1+ \beta W_{t+1}(\tau(\omega_1) , \hdots ,\tau(\omega_{n-2}) , p_{11},p_{01} )-
\beta W_{t+1}(p_{01}, \tau(\omega_1), \hdots , \tau(\omega_{n-2}), p_{11}) \geq 0 ~,\nn
\end{eqnarray}
where the last inequality is due to Lemma \ref{lem:bound}
(note that in that lemma we proved
$1+A \geq B$, which obviously implies $1+\beta A \geq \beta B$ for $0\leq \beta\leq 1$ that is used above).  This, together with the condition $x \geq y$, completes the proof.
\end{proof}

We are now ready to prove the main theorem.


{\em Proof of Theorem \ref{mainThm}:}
The basic approach is by induction on $t$.  The optimality of the myopic policy at time $t=T$ is obvious.  So the induction basis is established.  Now assume that the myopic policy is optimal for all times $t+1, t+2, \cdots, T-1$, and we will show that it is also optimal at time $t$.
By Lemma \ref{lem:nes-suf} this is equivalent to establishing the following
\begin{eqnarray}
W_t(\omega_1 , \hdots , \omega_{i-1} , \omega_{i+1} , \hdots , \omega_n , \omega_i) \leq  W_t(\omega_1, \hdots , \omega_n).
\end{eqnarray}

But we know from Lemmas \ref{lem:exchange} and \ref{lem:exchange-border} that,
\begin{eqnarray*}
&& W_t(\omega_1 , \hdots , \omega_{i-1} , \omega_{i+1} , \hdots , \omega_n , \omega_i) \leq  W_t(\omega_1 , \hdots , \omega_{i-1} , \omega_{i+1} , \hdots , \omega_i , \omega_n)  \nn \\
&\leq&
W_t(\omega_1 , \hdots , \omega_{i-1} , \omega_{i+1} , \hdots , \omega_i ,
 \omega_{n-1} , \omega_n)  \leq \hdots \leq  W_t(\omega_1, \hdots , \omega_n)~,
\end{eqnarray*}
where the first inequality is the result of Lemma~\ref{lem:exchange-border}, while the remaining inequalities
are repeated application of Lemma~\ref{lem:exchange},
completing the proof.
\hspace*{\fill}~\IEEEQED\par

\color{blue}
We would like to emphasize that from a technical point of view, Lemma \ref{lem:bound} is the key to the whole proof: it leads to Lemma \ref{lem:exchange-border}, which in turn leads to Theorem \ref{mainThm}.  While Lemma \ref{lem:exchange-border} was easy to conceptualize as a sufficient condition  to prove the main theorem, Lemma \ref{lem:bound} was much more elusive to construct and prove.  This, indeed, marks the main difference between the proof techniques used here vs. that used in our earlier work \cite{icc-08}: Lemma~\ref{lem:bound}
relies on a coupling argument instead of the convex analytic properties of the value function.
\color{black}

\section{The Case of $p_{11} < p_{01}$}\label{sec:counter}

In the previous section we showed that a myopic policy is optimal
if $p_{11} \geq p_{01}$.  In this section we examine what happens
when $p_{11} < p_{01}$, which corresponds to
the case when the Markovian channel state process exhibits a negative auto-correlation over a unit time.
\color{blue}
This is perhaps a case of less practical interest and relevance.  However, as we shall see this case presents a greater degree of technical complexity and richness than the previous case.  Specifically, we \color{black}
first show that when the number of channels
is three ($n=3$) or when the discount factor $\beta \leq \frac{1}{2}$,
the myopic policy remains optimal even for the
case of $p_{11} < p_{01}$ (the proof for two channels in this case
was given earlier in \cite{Zhao&etal:08TWC}).  We thus conclude that
the myopic policy is optimal for $n\leq 3$ or $\beta \leq 1/2$ regardless of
the transition probabilities.  We then present a counter example showing that
the the myopic policy is not optimal in general when $n\geq 4$
and $\beta > 1/2$.
In particular, our counter example is for a finite horizon
with $n=4$ and $\beta = 1$.

\subsection{$n=3$ or $\beta\leq \frac{1}{2}$}

We start by developing some results parallel to those presented
in the previous section for the case of $p_{11}\geq p_{01}$.

\begin{lemma}\label{thm:maina thm}
There exist $T$ n-variable polynomial functions of order $1$, denoted by $Z_t(), t=1, 2, \cdots, T$,
i.e., each function is linear in all the elements, and can be represented recursively in the following form:
\begin{eqnarray}\label{eq:Z}
Z_t(\bar \omega) := \omega_n (1+\beta Z_{t+1}( p_{11} , \tau(\omega_{n-1}),\hdots ,
\tau(\omega_{1})) )\nn\\
+(1-\omega_n) \beta Z_{t+1}(\tau(\omega_{n-1}),\hdots , \tau(\omega_{1}),p_{01} ) .
\end{eqnarray}
where $Z_T(\bar \omega)=\omega_n$.
\end{lemma}

\begin{cor}\label{col:Z}
\tcb{$Z_t(\bar\omega)$ given in (\ref{eq:Z}) represents the expected total reward of the myopic policy when $\bar\omega$
is ordered in increasing order of $\omega_i$.}
\end{cor}

\tcb{Similar to Corollary \ref{col:W}, the above result follows directly from the policy description given in Section \ref{sec:policy-myopic}. 
}

\ignore{
\begin{proof}
\tcm{The proof follows directly from the structure of the myopic policy established in \cite{Zhao&etal:08TWC}. Below we reproduce
the proof given in \cite{Zhao&etal:08TWC} for completeness.}

\tcb{We see from (\ref{eq:Z}) that $Z_t(\bar\omega)$ is the expected total reward of the following policy: we probe the $n$th channel (rightmost);
if it turns out to be ``1'', meaning for the next time step its probability of being available is $p_{11}$,
then we completely reverse the order of these channels, i.e., the $1$st channel is now in the $n$-th position, and so on; if it turns out to
be ``0'', then we keep this channel
in its original rightmost position (its availability probability for the next step is $p_{01}$), and reverse the order of the remaining $n-1$
channels.}

\tcb{To see that this is the myopic policy when $p_{11}< p_{01}$, we note that under the above policy at any time the list of information states
are increasingly ordered, and the policy always probes the $n$-th channel, the one with the largest probability of being available.
To see the increasing order, note that for any $0 \leq \omega \leq 1$, we have $p_{01} \leq \tau(\omega) \leq p_{11}$ when $p_{11} \geq p_{01}$.
Furthermore, under the assumption $p_{11} < p_{01}$, $\tau(\omega)$ is a decreasing function.  Therefore the channel reordering process of the
policy described above ensures that at each step the information states are ordered in increasing order.}
\end{proof}
} 

It follows that the function $Z_t$ also has the same linearity property
presented earlier, i.e.
\begin{eqnarray}
&&Z_t(\omega_1, \cdots, \omega_{n-2}, y, x) -
Z_t(\omega_1,\cdots,\omega_{n-2}, x, y)\nn\\
& =&
(x-y)(Z_t(\omega_1,\cdots,\omega_{n-2}, 0, 1)-
Z_t(\omega_{1},\cdots,\omega_{n-2}, 1, 0)) ~.
\end{eqnarray}
Similar results hold when we change the positions of $x$ and $y$.


In the next lemma and theorem we prove that the myopic
policy is still optimal when $p_{11} < p_{01}$ if
$n=3$ or $\beta \leq 1/2$ .
In particular, Lemma \ref{lem:n=3} below is the analogy of
Lemmas \ref{lem:exchange} and \ref{lem:exchange-border} combined.

\begin{lemma}\label{lem:n=3}
At time $t$ ($t=1, 2, \cdots, T$), for all $j\leq n-2$, we have the following
inequality for $\forall 1\geq x\geq y\geq 0$ if either $n=3$ or $\beta\leq 1/2$:
\begin{eqnarray}
\color{blue}
Z_t(\omega_1, \ldots, \omega_j, y, x, \omega_{j+3},\ldots, \omega_n)
&\geq&
Z_t(\omega_1, \ldots, \omega_j, x, y, \omega_{j+3},\ldots, \omega_n).
\color{black}
\end{eqnarray}
\end{lemma}

\begin{proof}
We prove this by induction on $t$.  The claim is obviously true for $t=T$.  Now suppose it's true for $t+1, \cdots, T-1$.
Due to the linearity property of $Z_t$,
\begin{eqnarray}
\color{blue}
\lefteqn{Z_t(\omega_1, \ldots, \omega_j, y, x, \omega_{j+3},\ldots, \omega_n) -
Z_t(\omega_1, \ldots, \omega_j, x, y, \omega_{j+3},\ldots, \omega_n)} \nonumber \\
& = & (x-y)\left(Z_t(\omega_1, \ldots, \omega_j, 0, 1, \omega_{j+3},\ldots, \omega_n) -
Z_t(\omega_1, \ldots, \omega_j, 1, 0, \omega_{j+3},\ldots, \omega_n) \right).
\color{black}
\end{eqnarray}
Thus it suffices to
show that
\color{blue}
$Z_t(\omega_1, \ldots, \omega_j, 0, 1, \omega_{j+3},\ldots, \omega_n) \geq
Z_t(\omega_1, \ldots, \omega_j, 1, 0, \omega_{j+3},\ldots, \omega_n) $.
\color{black}

We treat the case when $j < n-2$ and $j=n-2$ separately.  Indeed, without loss of generality, let $j=n-3$ (the proof follows exactly for all $j\leq n-3$ with more
lengthy notations).
At time $t$ we have
\begin{eqnarray*}
\lefteqn{ Z_t(\omega_1, \ldots, \omega_{n-3}, 0, 1, \omega_n) -
Z_t(\omega_1, \ldots, \omega_{n-3}, 1, 0, \omega_n) }\\
&=& \color{blue} \omega \beta (Z_{t+1}(p_{11}, p_{11}, p_{01}, \tau(\omega_{n-3}), \ldots, \tau(\omega_1))
- Z_{t+1}(p_{11}, p_{01}, p_{11}, \tau(\omega_{n-3}), \ldots, \tau(\omega_1) ))   \\
&+ & \color{blue} (1-\omega) \beta (Z_{t+1}(p_{11}, p_{01}, \tau(\omega_{n-3}), \ldots, \tau(\omega_1), p_{01})
- Z_{t+1}(p_{01}, p_{11}, \tau(\omega_{n-3}), \ldots, \tau(\omega_1), p_{01})) \\
\color{black}
&\geq& 0
\end{eqnarray*}
where the last inequality is due to the induction hypothesis.

Now we will consider the case when $j=n-2$.
\begin{eqnarray}
\lefteqn{ Z_t(\omega_1, \ldots, \omega_{n-2},0,1) - Z_t(\omega_1, \ldots, \omega_{n-2}, 1, 0) }\nonumber
\\
&=& 1+ \beta Z_{t+1}(p_{11}, p_{01}, \tau(\omega_{n-2}), \ldots, \tau(\omega_1)) - \beta Z_{t+1}(p_{11}, \tau(\omega_{n-2}), \ldots, \tau(\omega_1) , p_{01}).  \label{n3}
\end{eqnarray}

Next we show that if $\beta \leq 1/2$ or $n=3$ the right hand side of (\ref{n3})
is non-negative.

If $\beta \leq 1/2$, then
\begin{eqnarray*}
&& 1+ \beta Z_{t+1}(p_{11}, p_{01}, \tau(\omega_{n-2}), \ldots, \tau(\omega_1)) - \beta Z_{t+1}(p_{11}, \tau(\omega_{n-2}), \ldots, \tau(\omega_1) , p_{01})  \\
& & \hspace*{.3in} \geq  1 - \frac{\beta}{1-\beta} \geq 0.
\end{eqnarray*}

If $n=3$, then
\begin{eqnarray*}
\lefteqn{1+ \beta Z_{t+1}(p_{11}, p_{01}, \tau(\omega_1)) - \beta Z_{t+1}(p_{11}, \tau(\omega_1) , p_{01}) }
\\
&=& 1 + \beta (\tau(\omega_1) - p_{01}) (Z_{t+1}(p_{11}, 0, 1) - Z_{t+1}(p_{11}, 1, 0))  \\
&\geq& 1 - \beta (Z_{t+1}(p_{11}, 0, 1) - Z_{t+1}(p_{11}, 1, 0))\\
& \geq & 0
\end{eqnarray*}
where the first inequality is due to the fact that
$-1\leq \tau(\omega_1)-p_{01} \leq 0$ and the
last inequality is given by the induction hypothesis.
\end{proof}

\begin{theorem}\label{thm:n=3}
Consider Problem (P1). Assume that $p_{11}<p_{01}$.
The myopic policy is optimal for the case of $n=3$ and the case of $\beta\leq 1/2$ with arbitrary $n$. More precisely, for these two cases,  $\forall t$, $1\leq t \leq T$, we have
\begin{equation}
V_t(\bar\omega; a = j ) -
V_t(\bar\omega; a = i ) \geq 0,
\end{equation}
if $\omega_j \geq \omega_i$ for $i=1, \cdots, n$.
\end{theorem}

\begin{proof}
We prove by induction on $t$.  The optimality of the myopic policy at time
$t=T$ is obvious.  Now assume that the myopic policy is optimal for all
times $t+1, t+2, \cdots, T-1$, and we want to show that it is also optimal
at time $t$.  Suppose at time $t$ the channel probabilities are such that $\omega_n\geq \omega_i$ for $i=1, \cdots, n-1$. The myopic policy is optimal at time $t$ if and only if
probing $\omega_n$ followed by myopic probing is better than probing any
other channel followed by myopic probing. Mathematically, this means
\begin{eqnarray}
Z_t(\omega_1 , \hdots , \omega_{i-1} , \omega_{i+1} , \hdots , \omega_n , \omega_i) \leq
Z_t(\omega_1, \hdots , \omega_n), \hspace*{.4in}
\mbox{ for all }  \omega_1 \leq \omega_i \leq \omega_n. \nn
\end{eqnarray}
But this is a direct consequence of Lemma~\ref{lem:n=3},
completing the proof.
\end{proof}

\subsection{A $4$-channel Counter Example}


The following example shows that the myopic policy is not, in general,
optimal for $n\geq 4$ when $p_{11}<p_{01}$.

\begin{ex}
Consider an example with the following parameters:
$p_{01} = 0.9, p_{11} = 0.1, \beta=1$, and $\bar\omega = [.97, .97, .98, .99]$.
Now compare the following two policies at time $T-3$: play myopically (I), or play the $.98$ channel first, followed by the myopic policy (II).  Computation reveals that
\begin{eqnarray*}
&& V_{T-3}^{I} (.97, .97, .98, .99) = 2.401863 \\
&<& V_{T-3}^{II} (.97, .97, .98, .99) = 2.402968
\end{eqnarray*}
which shows that the myopic policy is not optimal in this case.
\end{ex}

It remains an interesting question as to whether such counter examples exist in the case when the initial condition is such that all channel are in the good state with the stationary probability.


\section{Infinite Horizon}\label{sec:infinite}

Now we consider extensions of results in Sections~\ref{sec:optimal}
and \ref{sec:counter}
to (P2) and (P3), i.e.,
to show that the myopic policy is also optimal for (P2) and (P3) under
the same conditions.
Intuitively, this holds due to the fact that the  stationary optimal policy of the finite horizon problem
is independent of the horizon as well as the discount factor. Theorems~\ref{thm:P2} and \ref{thm:P3}
below concretely establish this.

We point out that the proofs of Theorems~\ref{thm:P2} and \ref{thm:P3} do not rely on any additional
assumptions other than the optimality of the myopic
policy for (P1).
Indeed, if the optimality of the myopic policy for (P1)
can be established under weaker conditions,
Theorems~3 and 4 can be readily invoked to
establish its optimality under the same weaker condition for
(P2) and (P3), respectively.



\begin{theorem}
If myopic policy is optimal
for (P1), it is also optimal for (P2) for $0 \leq \beta <1$.  Furthermore, its value function is
the limiting  value function of (P1) as the time horizon goes to
infinity, i.e., we have $\max_{\pi} J_{\beta}^{\pi}(\bar\omega) =
\lim_{T\rightarrow\infty}\max_{\pi} J_{T}^{\pi}(\bar\omega)$.
\label{thm:P2}
\end{theorem}

\begin{proof}
We first use the bounded convergence theorem (BCT) to establish the fact that
under any deterministic stationary Markov policy $\pi$, we have
$J_{\beta}^{\pi}(\bar\omega) = \lim_{T\rightarrow\infty} J_{T}^{\pi}(\bar\omega)$.
We prove this by noting that
\begin{eqnarray}
J_{\beta}^{\pi}(\bar\omega) &=& E^{\pi} [ \lim_{T\rightarrow\infty}
\sum_{t=1}^{T} \beta^{t-1} R_{\pi(t)}(\bar\omega(t)) | \bar\omega(1) = \bar\omega]
\nonumber\\
&=& \lim_{T\rightarrow\infty} E^{\pi}[ \sum_{t=1}^T \beta^{t-1} R_{\pi(t)}(\bar\omega(t)) | \bar\omega(1) = \bar\omega] \nonumber\\
&=& \lim_{T\rightarrow\infty} J_{T}^{\pi}(\bar\omega)
\end{eqnarray}
where the second equality is due to BCT for
 $\sum_{t=1}^{T} \beta^{t-1} R_{\pi(t)}(\bar\omega(t)) \leq \frac{1}{1-\beta}$.
This proves the second part of the theorem by noting that due to the finiteness
of the action space, we can interchange maximization and limit.

Let $\pi^*$ denote the myopic policy.
We now establish the optimality of $\pi^{*}$ for (P2).
%
%
From Theorem~1, we know:
\begin{eqnarray}
J^{\pi^*}_T (\bar\omega) &=& \max_{a=i} \left\{ \omega_i +  \beta
\omega_i J^{\pi^*}_{T-1} \left( \Tc\left( {\bar\omega}, i | 1\right)\right) \right. \nn \\
&& \hspace*{.75in} + \beta \left. 
(1-\omega_i)  J^{\pi^*}_{T-1} \left( \Tc\left( {\bar\omega}, i | 0 \right) \right) \right\}. \nn
\end{eqnarray}
Taking limit of both sides, we have
\begin{eqnarray}\label{DP-myopic}
J^{\pi^*}_\beta (\bar\omega) && = \max_{a=i} \left\{ \omega_i +  \beta
\omega_i J^{\pi^*}_{\beta} \left( \Tc\left( {\bar\omega}, i | 1\right)\right) \right. \nn \\
&& \hspace*{.75in} + \beta\left.
(1-\omega_i)  J^{\pi^*}_{\beta} \left( \Tc\left( {\bar\omega}, i | 0\right) \right) \right\}.
\end{eqnarray}
Note that (\ref{DP-myopic}) is nothing but the dynamic programming equation for
the infinite horizon discounted reward problem given in \eqref{DP-discount}.
From the uniqueness of the dynamic programming solution, then, we have
\[
 J^{\pi^*}_\beta(\bar\omega) = V_\beta (\bar\omega)  = \max_{\pi} J_{\beta}^{\pi}({\bar\omega})
\]
hence, the optimality of the myopic policy.


\end{proof}

\begin{theorem}
Consider (P3) with the expected average reward and under the ergodicity assumption
$| p_{11}-p_{00} | < 1$.  Myopic
policy is optimal for problem (P3) if it is optimal for (P1).
\label{thm:P3}
\end{theorem}

\begin{proof}
We consider the infinite horizon discounted cost for $\beta <1$ under the optimal policy denoted by $\pi^*$:
\begin{eqnarray} \label{starting}
\lefteqn{J^{\pi^*}_\beta (\bar\omega) = \max_{a=i} \left\{ \omega_i + \beta
\omega_i J^{\pi^*}_\beta \left( \Tc\left( {\bar\omega}, i | 1\right)\right) \right. } \nn \\
&& \hspace*{.7in} \left. + \beta
(1-\omega_i)  J^{\pi^*}_\beta \left( \Tc\left( {\bar\omega}, i | 0\right) \right) \right\}.
\end{eqnarray}
This can be written as
\begin{eqnarray}
\lefteqn{(1 - \beta) J^{\pi^*}_\beta (\bar\omega) } \nn \\
&  = &   \max_{a=i} \left\{ \omega_i + \beta
\omega_i \left[ J^{\pi^*}_\beta \left( \Tc\left( {\bar\omega}, i | 1\right)\right) - J^{\pi^*}_\beta (\bar\omega) \right] \right. \nn \\
& & \left. + \beta
(1-\omega_i)  \left[ J^{\pi^*}_\beta \left( \Tc\left( {\bar\omega}, i | 0\right) \right)
- J^{\pi^*}_\beta (\bar\omega) \right] \right\}. \nn
\end{eqnarray}

Notice that the boundedness of the reward function and compactness of information state
 implies that the sequence of $\{(1 - \beta) J^{\pi^*}_\beta (\bar\omega)\}$
 is bounded, \color{blue} i.e. for all $0\leq \beta \leq 1$,
 \begin{equation}\label{bdd-seq}
(1 - \beta) J^{\pi^*}_\beta (\bar\omega) \leq 1. 
 \end{equation}
 Also, applying Lemma 2 from \cite{icc-08} (which provides an upper bound on 
the difference in value functions between taking two different actions followed
by the optimal policy) and noting that 
$-1 < p_{11}-p_{00} < 1$, we have that there
exists some positive constant $K := \frac{1}{1 - |p_{11}-p_{01}|}$ such that
\begin{equation}\label{bdd-seq2}
\left | J^{\pi^*}_{\beta} \left( \Tc\left( {\bar\omega}, i | 0\right) \right)
- J^{\pi^*}_\beta (\bar\omega) \right | \leq K.
\end{equation}

\color{black} 
By Bolzano-Weierstrass theorem, (\ref{bdd-seq}) and (\ref{bdd-seq2}) guarantee the
existence of a converging sequence $\beta_k \rightarrow 1$ such that
 \begin{eqnarray}
&&  \color{blue} \lim_{k \rightarrow \infty} (1 - \beta_k) J^{\pi^*}_{\beta_k} (\bar\omega^*)  :=J^*,  \label{proc1}\\
\mbox{and} && \color{blue}
\lim_{k \rightarrow \infty} \left[J^{\pi^*}_{\beta_k} (\bar\omega) -
J^{\pi^*}_{\beta_k} (\bar\omega^*) \right] := h^{\pi^*}(\bar\omega) ~, \label{h}
 \end{eqnarray}
where $\omega^*_i: = \frac{p_{01}}{1- p_{11}+ p_{01}}$ is the steady-state belief (the limiting belief when channel $i$ is not sensed for a long time).

As a result, (\ref{proc1}) can be written as
\begin{eqnarray}
J^*  
 = \lim_{k \rightarrow \infty}  \left\{ (1 - \beta_k) J^{\pi^*}_{\beta_k} (\bar\omega^*)
+  (1 - \beta_k) \left[J^{\pi^*}_{\beta_k} (\bar\omega) -
J^{\pi^*}_{\beta_k} (\bar\omega^*)\right] \right\}.\nn
\end{eqnarray}

In other words,
\begin{eqnarray}
\lefteqn{J^* 
=  \lim_{k \rightarrow \infty}  \max_{a=i} \left\{ \omega_i + \beta_k
\omega_i \left[ J^{\pi^*}_{\beta_k} \left( \Tc\left( {\bar\omega}, i | 1\right)\right) \right. \right.} \nn \\
&  & \left. - J^{\pi^*}_{\beta_k} (\bar\omega) \right]
 + \beta_k \left. 
(1-\omega_i)  \left[ J^{\pi^*}_{\beta_k} \left( \Tc\left( {\bar\omega}, i | 0\right) \right)
- J^{\pi^*}_{\beta_k} (\bar\omega) \right] \right\}. \nn 
\end{eqnarray}
From (\ref{h}), we can write this as
\begin{eqnarray}  \label{optimality}
\lefteqn{J^* + h^{\pi^*} (\bar\omega) =  \max_{a=i} \left\{ \omega_i +
\omega_i h^{\pi^*} \left( \Tc\left({\bar\omega}, i | 1\right)\right) + \right. }\nn \\
& &  \left.
 \hspace*{.9in} (1-\omega_i)  h^{\pi^*} \left( \Tc\left({\bar\omega}, i | 0\right) \right) \right\}.
\end{eqnarray}

Note that  (\ref{optimality}) is nothing but the DP equation as given
by (\ref{DP-avgCost}). In addition, we know that the immediate reward as well 
as function $h$ are both bounded by $\max(1,K)$. 
This implies that $J^*$ is the maximum average reward, i.e.\
$J^* = \max_{\pi} J_{\infty}^{\pi}({\bar\omega(t)})$ {\color{blue}{(see \cite[Theorems 6.1-6.3]{marcus-survey})}}.

On the other hand, we know from Theorem \ref{thm:P2} that the myopic policy is optimal for (P2) if it is for (P1), and thus we can take $\pi^*$ in (\ref{starting}) to be the myopic policy.  Rewriting (\ref{starting}) gives the following: 
\begin{eqnarray*}
\lefteqn{J^{\pi^*}_\beta (\bar\omega) = \omega_{\pi^*(\bar\omega)} + \beta
\omega_{\pi^*(\bar\omega)} J^{\pi^*}_\beta \left( \Tc\left( {\bar\omega}, {\pi^*(\bar\omega)} | 1\right)\right) } \nn \\
&&  \hspace*{.5in}  + \beta
(1-\omega_{\pi^*(\bar\omega)})  J^{\pi^*}_\beta \left( \Tc\left( {\bar\omega}, {\pi^*(\bar\omega)} | 0\right) \right) ~. 
\end{eqnarray*}
Repeating steps (\ref{proc1})-(\ref{optimality}) we arrive at the following:
\begin{eqnarray}
\lefteqn{J + h^{\pi^*} (\bar\omega) =  \omega_{\pi^*(\bar\omega)} +
\omega_{\pi^*(\bar\omega)} h^{\pi^*} \left( \Tc\left(\bar\omega, {\pi^*(\bar\omega)} | 1\right)\right) + }
\nn \\
& &  \hspace*{.7in}
(1-\omega_{\pi^*(\bar\omega)})  h^{\pi^*} \left( \Tc\left(\bar\omega, {\pi^*(\bar\omega)} | 0\right) \right),
\end{eqnarray}
which {\color{blue}{shows that 
$(J^*, h^{\pi^*}, \pi^*)$ is a canonical triplet \cite[Theorems 6.2]{marcus-survey}. This, 
together with boundedness of $h^{\pi^*}$ and immediate reward,}} implies that
the myopic policy $\pi^*$ is optimal for (P3) {\color{blue}{\cite[Theorems 6.3]{marcus-survey}}}. 
\end{proof}

\section{Discussion and Related Work} \label{sec:discussion}



\tcb{The problem studied in this paper may be viewed as a special case of a class of MDPs known as the {\em restless bandit problems} \cite{whittle}.  In this class of problems, $N$ controlled Markov chains (also called {\em projects} or {\em machines}) are activated (or played) one at a time.  A machine when activated generates a state dependent reward and transits to the next state according to a Markov rule.   A machine not activated transits to the next state according to a (potentially different) Markov rule.  The problem is to decide the sequence in which these machines are activated so as to maximize the expected (discounted or average) reward over an infinite horizon.  
To put our problem in this context, each channel corresponds to a machine, and a channel is activated when it is probed, and its information state goes through a transition depending on the observation and the underlying channel model.  When a channel is not probed, its information state goes through a transition solely based on the underlying channel model \footnote{
The standard definition of bandit problems typically assumes finite or countably infinite state spaces. While our problem can potentially have an uncountable state space, it is nevertheless countable for a given initial state.  This view has been taken throughout the paper.}. 
}

In the case that a machine stays frozen in its current state when not played, the problem reduces to the {\em multi-armed bandit problem}, a class of problems solved by Gittins in his 1970 seminal work \cite{gittins}.  Gittins showed that there exists an {\em index} associated with each machine that is solely a function of that individual machine and its state, and that playing the machine currently with the highest index is optimal.  This index has since been referred to as the {\em Gittins index} due to Whittle \cite{whittle-80}.  
The remarkable nature of this result lies in the fact that it essentially decomposes the $N$-dimensional problem into $N$ 1-dimensional problems, as an index is defined for a machine independent of others. 
The basic model of multi-armed bandit has been used previously in the context of channel access and cognitive radio networks. For example, in \cite{motamedi-07}, Bayesian learning was used to estimate the probability of a channel being available, and
the Gittins indices, calculated based on such estimates (which were only updated when a channel is observed and used, thus giving rise to a multi-armed bandit formulation rather than a restless bandit formulation), were used for channel selection.

On the other hand, relatively little is known about the structure of the optimal policies for the restless bandit problems in general.  
It has been shown that the Gittins index policy is not in general optimal in this case \cite{whittle}, and that this class of problems is PSPACE-hard in general \cite{tsitsiklis}.
Whittle, in \cite{whittle}, proposed a Gittins-like index (referred to as the Whittle's index policy), 
shown to be optimal under a constraint on the {\em average} number of machines that can be played at a given time, 
and asymptotically optimal under certain limiting regimes \cite{weber}.
%
There has been a large volume of literature in this area, including various approximation algorithms, see for example \cite{bertsimas-index} and \cite{nino-mora} for near-optimal heuristics, as well as conditions for certain policies to be optimal for
special cases of the restless bandit problem, see e.g., \cite{lott-00,ehsan-twc}.
\tcb{
The nature of the results derived in the present paper is similar to that of \cite{lott-00,ehsan-twc} in spirit.  That is, we have shown that for this special case of the restless bandit problem an index policy is optimal under certain conditions. 
For the indexability (as defined by Whittle \cite{whittle}) of this problem, see 
\cite{Liu&Zhao:08SDR}.} 

Recently Guha and Munagala \cite{guha,guha2} studied a class of problems referred to as the {\em feedback multi-armed bandit} problems.  This class is very similar to the restless bandit problem studied in the present paper, with the difference that channels may have different transition probabilities (thus this is a slight generalization to the one studied here). 
%
While we identified conditions under which a simple greedy index policy is optimal in the present paper, 
Guha and Munagala in \cite{guha,guha2} looked for provably good approximation algorithms.
In particular, they derived a $2+\epsilon$-approximate policy using a duality-based technique. 

\section{Conclusion} \label{sec:conclusion}

The general problem of opportunistic sensing and access arises in
many multi-channel communication contexts. For cases where the stochastic evolution  of channels can be modelled as i.i.d. two-state Markov chains, we showed that a simple and robust myopic policy is optimal for \tcb{the finite and infinite horizon discounted reward criteria as well as the infinite horizon average reward criterion, when the state transitions are positively correlated over time.  When the state transitions are negatively correlated, we showed that the same policy is optimal when the number of channels is limited to 2 or 3, and presented a counterexample for the case of 4 channels.} 

%



\bibliography{osa-IT}

\end{document}